\documentclass{acm_proc_article-sp}

\usepackage{url}
\makeatletter
\newif\if@restonecol
\makeatother
\usepackage[ruled,vlined, linesnumbered]{algorithm2e}
\usepackage{amsmath,amsfonts,latexsym,amssymb}
\usepackage{tikz}
\usepackage{url}
\usepackage{lhelp}
\def\A {\ensuremath{\mathbb{A}}}
\def\B {\ensuremath{\mathbb{B}}}

\def\K {\ensuremath{\mathbf{k}}}
\def\L {\ensuremath{\mathfrak{L}}}

\def\Q {\ensuremath{\mathbb{Q}}}

\newcommand{\T}{\mathfrak{T}}

\def\KK {\ensuremath{\mathbf{K}}}

\newtheorem{Theorem}{Theorem}
\newtheorem{Proposition}{Proposition}
\newtheorem{Definition}{Definition}
\newtheorem{Lemma}{Lemma}
\newtheorem{Corollary}{Corollary}

\newcommand{\init}[1]{\mbox{{\rm init}$(#1)$}}
\newcommand{\iter}[1]{\mbox{{\rm res}$(#1)$}}
\newcommand{\mdeg}[1]{\mbox{{\rm mdeg}$(#1)$}}
\newcommand{\mvar}[1]{\mbox{{\rm mvar}$(#1)$}}
\newcommand{\prem}[1]{\mbox{{\rm prem}$(#1)$}}
\newcommand{\pquo}[1]{\mbox{{\rm pquo}$(#1)$}}
\newcommand{\rank}[1]{\mbox{{\rm rank}$(#1)$}}
\newcommand{\res}[1]{\mbox{{\rm res}$(#1)$}}
\newcommand{\sat}[1]{\mbox{{\rm sat}$(#1)$}}

\newcommand{\head}[1]{\mbox{{\rm head}$(#1)$}}
\newcommand{\tail}[1]{\mbox{{\rm tail}$(#1)$}}
\newcommand{\coeff}[1]{\mbox{{\rm coeff}$(#1)$}}
\renewcommand{\min}[1]{\mbox{{\rm min}$(#1)$}}
\renewcommand{\max}[1]{\mbox{{\rm max}$(#1)$}}

\newcommand{\Extend}[1]{\mbox{{\sf Extend}$(#1)$}}
\newcommand{\Regularize}[1]{\mbox{{\sf Regularize}$(#1)$}}

\newcommand{\Intersect}[1]{\mbox{{\sf Intersect}$(#1)$}}
\newcommand{\Triangularize}[1]{\mbox{{\sf Triangularize}$(#1)$}}

\newcommand{\IntersectAlgebraic}[1]{\mbox{{\sf IntersectAlgebraic}$(#1)$}}
\newcommand{\SubresultantChain}[1]{\mbox{{\sf SubresultantChain}$(#1)$}}
\newcommand{\resultant}[1]{\mbox{{\sf resultant}$(#1)$}}
\newcommand{\RegularGcd}[1]{\mbox{{\sf RegularGcd}$(#1)$}}

\newcommand{\IntersectFree}[1]{\mbox{{\sf IntersectFree}$(#1)$}}
\newcommand{\CleanChain}[1]{\mbox{{\sf CleanChain}$(#1)$}}

\newcommand{\Squarefree}[1]{\mbox{{\sf Squarefree}$(#1)$}}

\newcommand{\SquarefreePart}[1]{\mbox{{\sf SquarefreePart}$(#1)$}}

\newcommand{\Maple}{{\sc  Maple}}

\newcommand{\RegularChains}{{\tt  Regu\-lar\-Chains}}

\newcommand{\GCD}{{\small GCD}}

\def\x {\ensuremath{\mathbf{x}}}

\newcommand{\dpol}[1]{\mbox{{\rm dpol}$(#1)$}}

\newcommand{\red}[1]{\mbox{{\rm red}$(#1)$}}
\renewcommand{\gcd}[1]{\mbox{{\rm gcd}$(#1)$}}

\newcommand{\p}{\mathfrak{p}}
\def\DD {\ensuremath{\mathbb{D}}}
\def\LL {\ensuremath{\mathbb{L}}}

\newcommand{\citep}[1]{\mbox{ \cite{#1} }}
\newif\ifcomment
\commentfalse

\begin{document}

\title{Algorithms for Computing Triangular Decompositions of Polynomial Systems}

\numberofauthors{2} 
\author{
\alignauthor
Changbo Chen\\
       \affaddr{ORCCA, University of Western Ontario (UWO)} \\
       \affaddr{London, Ontario, Canada} \\
       \email{cchen252@csd.uwo.ca}
\alignauthor
Marc Moreno Maza\\
       \affaddr{ORCCA, University of Western Ontario (UWO)} \\
       \affaddr{London, Ontario, Canada} \\
       \email{moreno@csd.uwo.ca}
}
\date{\today}

\maketitle

\begin{abstract}
We propose new algorithms for computing triangular
decompositions of polynomial systems incrementally.
With respect to previous works, our improvements
are based on a {\em weakened} notion of a polynomial GCD 
   modulo a regular chain, which permits
to greatly simplify and optimize the sub-algorithms.
Extracting common work from similar expensive computations
is also a key feature of our algorithms.
In our experimental results the implementation
of our new algorithms, realized with the
   {\RegularChains} library in {\Maple}, 
   outperforms solvers with similar specifications by several orders 
   of magnitude on sufficiently difficult problems.
\end{abstract}

\section{Introduction}

The Characteristic Set Method~\cite{Wu87} of Wu has freed Ritt's decomposition
from polynomial factorization, opening the door to a variety of 
discoveries in polynomial system solving. 
In the past two decades the work of Wu 
has been extended to more powerful decomposition algorithms
and applied to different types
of polynomial systems or decompositions: 
differential systems~\cite{BLOP95,Hubert00},
difference systems~\cite{GaoHoevenLuoYuan2009},
real parametric systems~\cite{yhx01},
primary decomposition~\cite{ShiYo96},
cylindrical algebraic decomposition~\cite{CMXY09}.
Today, triangular decomposition algorithms
provide back-engines for 
computer algebra system front-end solvers,
such as {\sc Maple}'s {\tt solve} command.

Algorithms computing triangular decompositions of polynomial systems
can be classified in several ways. One can first consider the relation
between the input system $S$ and the output triangular systems 
$S_1, \ldots, S_e$. From that perspective, two types of decomposition 
are essentially different:  those for which $S_1, \ldots, S_e$ encode
all the points of the zero set $S$ (over the algebraic closure 
of the coefficient field of $S$) and those for which $S_1, \ldots, S_e$ 
represent only the ``generic zeros'' 
of the irreducible components of $S$.

One can also classify triangular decomposition algorithms by the
algorithmic principles on which they rely. From this other angle,
two types of algorithms are essentially different:  those
which proceed {\em by variable elimination}, that is, by reducing 
the solving of a system in $n$ unknowns to that of a system in
$n-1$ unknowns and those which proceed {\em incrementally}, that is,
by reducing the solving of a system in $m$ equations to that of a 
system in $m-1$ equations.

The Characteristic Set Method and 
the algorithm in~\cite{Wang00b} belong to the
first type in each classification.
Kalkbrener's algorithm~\cite{Kalk93}, which is an elimination
me\-thod solving in the sense of the ``generic zeros'', has brought
efficient techniques, based on the concept of a {\em regular chain}.
Other works~\cite{Laz91a,moreno00}  
on triangular decomposition algorithms 
focus on incremental solving.
This principle is quite attractive, since it allows
to control the properties and size of the intermediate
computed objects. It is used in other areas of polynomial
system solving such as the probabilistic algorithm
of Lecerf~\cite{Lecerf03} based on lifting fibers
and the numerical method of 
Sommese, Verschelde, Wample~\cite{AAG2008}
based on diagonal homotopy.

Incremental algorithms for triangular decomposition
   rely on a procedure for computing the intersection of an 
   hypersurface and the quasi-component of a regular chain.
   Thus, the input of this operation can be regarded
   as well-behaved geometrical objects.
   However, known algorithms, namely the one of 
   Lazard~\cite{Laz91a} and the one of the second author~\cite{moreno00}
   are quite involved and difficult to analyze and optimize.

In this paper, we revisit this intersection operation. 
Let $R = {\K}[x_1, \ldots, x_n]$ be the ring 
of multivariate polynomials with coefficients in {\K}
and ordered variables $\x=x_1 < \cdots < x_n$.
Given a polynomial $p \in R$
and a regular chain $T \subset {\K}[x_1, \ldots, x_n]$,
the function call $\Intersect{p, T, R}$ returns 
regular chains $T_1,\ldots,T_e \subset {\K}[x_1, \ldots, x_n]$
such that we have:
$$
V(p)\cap W(T)\subseteq W(T_1)\cup\cdots\cup W(T_e)\subseteq V(p)\cap\overline{W(T)}.
$$
(See Section~\ref{sec:regularchains} for the notion 
of a regular chain and related concepts and notations.)
We exhibit an algorithm for computing $\Intersect{p, T, R}$
which is conceptually simpler and practically much more
efficient than those of~\cite{Laz91a,moreno00}.
Our improvements result mainly from two new ideas.

\smallskip\noindent{\small \bf Weakened notion of polynomial GCDs modulo regular chain.}
Modern algorithms for triangular decomposition rely
implicitly or explicitly on a notion of GCD for 
univariate polynomials over an
arbitrary commutative ring.
A formal definition was proposed in~\cite{moreno00} 
(see Definition~\ref{Definition:regular-gcd})
and applied to residue class rings of the form 
${\A} = {\K}[\x] / \sat{T}$
where $\sat{T}$ is the saturated ideal of the regular chain $T$.
A modular algorithm for computing these GCDs appears
in~\cite{LiMorenoPan09}:
if $\sat{T}$ is known to be radical, the performance
(both in theory and practice) of this algorithm
are very satisfactory whereas if $\sat{T}$ is not
radical, the complexity of the algorithm increases
substantially w.r.t. the radical case.
In this paper, the ring {\A}
will be of the form $\K[\x]/\sqrt{\sat{T}}$
while our algorithms will not need to compute a basis nor
a characteristic set of $\sqrt{\sat{T}}$.
For the purpose of polynomial system solving
(when retaining the multiplicities of zeros is not required)
this weaker notion of a polynomial GCD is clearly sufficient.
In addition, this leads us to a very
simple procedure for computing such GCDs, 
see Theorem~\ref{Theorem:regulargcd}.
To this end, we rely on the {\em specialization property
of subresultants}.
Appendix~\ref{app:subresultantchain} reviews this property
and provides corner cases for which we could 
not find a reference in the literature.

\smallskip\noindent{\small \bf Extracting common work from similar computations.}
Up to technical details, 
if $T$ consists of a single polynomial $t$ whose main variable 
is the same as $p$, say $v$,
computing $\Intersect{p, T, R}$ can be achieved by successively
computing
\begin{itemize}
\item[$(s_1)$] the resultant $r$ of $p$ and $t$ w.r.t. $v$,
\item[$(s_2)$] a regular GCD of $p$ and $t$ modulo
      the squarefree part of $r$.
\end{itemize}
Observe that Steps $(s_1)$ and $(s_2)$ reduce essentially
to computing the subresultant chain of $p$ and $t$ w.r.t. $v$.
The algorithms of Section~\ref{sec:incremental} extend
this simple observation for computing $\Intersect{p, T, R}$ 
with an arbitrary regular chain.
In broad terms, the intermediate polynomials 
computed during the ``elimination phasis''
of $\Intersect{p, T, R}$  are recycled for
performing the ``extension phasis'' at essentially no cost.

The techniques developed for $\Intersect{p, T, R}$ 
are applied to other key sub-algorithms, such as:
\begin{itemize}
\item the regularity test of a polynomial modulo the saturated
      of a regular chain, see Section~\ref{sec:incremental},
\item the squarefree part of a regular chain, 
      see Appendix~\ref{sec:squarefree}.
\end{itemize}
The primary application of the operation {\sf Intersect}
is to obtain triangular decomposition encoding all 
the points of the zero set of the input system.
However, we also derive from it in Section~\ref{sec:Kalkbrener}
an algorithm
computing triangular decompositions in the sense of Kalkbrener.

\smallskip\noindent{\small \bf Experimental results.}
We have implemented the algorithms  presented in this 
paper within the {\RegularChains} library in {\Maple},
leading to a new implementation of the {\sf Triangularize} command.
In Section~\ref{sec:Implementation}, we report on various 
benchmarks.
This new version of  {\sf Triangularize}  outperforms
the previous ones 
(based on~\cite{moreno00}) 
by several orders 
   of magnitude on sufficiently difficult problems.
Other {\Maple} commands or packages 
for solving polynomial systems
(the {\tt WSolve} package,
 the {\tt Groebner:-Solve} command 
  and the {\tt Groebner:-Basis} command for a lexicographical term order)
are also outperformed by the implementation of the algorithms
presented in this paper both in terms of running time
and, in the case of engines based on Gr\"obner bases,
in terms of output size.

\section{Regular chains}  %
\label{sec:regularchains}

We review hereafter the notion of a
regular chain and its related concepts.
Then we state basic properties
(Propositions~\ref{Proposition:prem},
\ref{Proposition:equal-dim},
\ref{Proposition:radsatprem},
\ref{Proposition:equivalence},
and Corollaries \ref{Corollary:sqrtsat},
\ref{Corollary:regular})
of regular chains, which are at the core
of the proofs of the algorithms of Section~\ref{sec:incremental}.

Throughout this paper, $\K$ is a field, 
$\KK$ is the algebraic closure of $\K$
and $\K[\x]$ denotes the ring of  polynomials over $\K$, 
with ordered variables $\x= x_1 < \cdots < x_n$. 
Let $p\in {\K}[\x]$.

\smallskip\noindent{\small \bf Notations for polynomials.}
If $p$ is not constant, then the greatest variable appearing in $p$ is
called the {\em main variable} of $p$, denoted by $\mvar{p}$.
Furthermore, the leading coefficient, the degree, the
leading monomial, the leading term and 
the reductum of $p$, regarded as a univariate polynomial in
$\mvar{p}$, are called respectively the {\em initial}, the {\em main degree},
the {\em rank},  the {\em head} and the {\em tail} of $p$; 
they are denoted by $\init{p}$, $\mdeg{p}$, $\rank{p}$, 
$\head{p}$ and $\tail{p}$ respectively.
Let $q$ be another polynomial of $\K[\x]$.
If $q$ is not constant, 
then we denote by $\prem{p, q}$ and $\pquo{p, q}$
the pseudo-remainder and the pseudo-quotient of $p$ by $q$
as univariate polynomials in $\mvar{q}$.
We say that $p$ is less than q and write $p\prec q$ 
if either $p\in\K$ and $q\notin\K$ or 
both are non-constant polynomials such that
$\mvar{p}<\mvar{q}$ holds,
 or $\mvar{p}=\mvar{q}$ and $\mdeg{p}<\mdeg{q}$ both hold.
We write $p\sim q$ if neither $p\prec q$ nor $q\prec p$ hold.

\smallskip\noindent{\small \bf Notations for polynomial sets.}
Let $F\subset \K[\x].$
We denote by $\langle F \rangle$ the ideal generated by $F$ in
${\K}[\x]$.
For an ideal ${\cal I} \subset \K[\x]$,
we denote by ${\rm dim}({\cal I})$ its dimension.
A polynomial is {\em regular} modulo ${\cal I}$ if it is neither
zero, nor a zerodivisor modulo ${\cal I}.$
Denote by $V(F)$ the {\em zero set} 
(or algebraic variety) of $F$ in ${\KK}^n.$
Let $h \in \K[\x].$
The {\em{saturated ideal}} of ${\cal I}$   w.r.t. $h$, 
denoted by ${\cal I} :h^{\infty}$, is the ideal 
$\{q\in \K[\x]\mid \exists m\in\mathbb{N}\text{ s.t. }
h^mq\in {\cal I}\}$.

\smallskip\noindent{\small \bf Triangular set.}
Let $T \subset {\K}[\x]$ be a {\em triangular set},
that is, a set of non-constant polynomials with pairwise
distinct main variables.
The set of
main variables and the set of ranks of the polynomials in $T$
are denoted by $\mvar{T}$ and $\rank{T}$, respectively.
A variable in $\x$ is called
{\em algebraic} w.r.t. $T$ if it belongs to \mvar{T}, otherwise
it is said {\em free} w.r.t. $T$. 
For $v\in\mvar{T}$, denote by $T_v$ the polynomial in $T$
with main variable $v$.
For $v \in \x$,
we denote by $T_{<v}$ (resp. $T_{\geq v}$) the set of 
polynomials $t \in T$ such that $\mvar{t} < v$ (resp. $\mvar{t}\geq v$) holds.
Let $h_T$ be the product
of the initials of the polynomials in $T$.
We denote by \sat{T} the {\em saturated ideal} of $T$
defined as follows:
if $T$ is empty then \sat{T} is the trivial
ideal $\langle 0 \rangle$, otherwise it is the ideal
$\langle T \rangle:h_{T}^{\infty}$.
The {\em quasi-component} $W(T)$ of $T$
is defined as $V(T) \setminus V(h_T)$.
Denote $\overline{W(T)} = V(\sat{T})$ as the Zariski
closure of $W(T)$.
For $F\subset \K[\x]$,  
we write  $Z(F, T):=V(F)\cap W(T)$.

\smallskip\noindent{\small \bf Rank of a triangular set.}
Let $S \subset \K[\x]$ be another triangular set.
We say that $T$ has smaller rank than $S$ and we write
$T\prec S$ if there exists $v\in\mvar{T}$ such that
$\rank{T_{<v}}=\rank{S_{<v}}$ holds and: $(i)$ either $v\notin\mvar{S}$;
$(ii)$ or $v\in\mvar{S}$ and $T_{v}\prec S_{v}$.
We write $T\sim S$ if $\rank{T}=\rank{S}$.

\smallskip\noindent{\small \bf Iterated resultant.}
Let $p,q \in \K[\x]$. 
Assume $q$ is nonconstant and let $v=\mvar{q}$.
We define $\res{p, q, v}$ as follows:
if the degree ${\deg}(p,v)$ of $p$ in $v$ is null, 
then $\res{p, q, v}=p$; otherwise $\res{p, q, v}$
is the resultant of $p$ and $q$ w.r.t. $v$.
Let $T$ be a triangular set of $\K[\x]$.
We define $\res{p, T}$ by induction:
if $T=\varnothing$, then $\res{p, T}=p$;
otherwise let $v$ be greatest variable appearing in $T$, then 
$\res{p, T}=\res{\res{p, T_v, v}, T_{<v}}$.

\smallskip\noindent{\small \bf Regular chain.}
A triangular set $T \subset \K[\x]$ is a {\em regular chain} 
if: $(i)$ either $T$ is empty; 
$(ii)$ or $T \setminus \{ T_{\rm max} \}$ is a regular chain,
where $T_{\rm max}$ is the polynomial in $T$ with maximum rank, 
and the initial of $T_{\rm max}$ is regular w.r.t. 
\sat{T \setminus \{ T_{\rm max} \}}.
The empty regular chain is simply denoted by $\varnothing$.

\smallskip\noindent{\small \bf Triangular decomposition.}
Let $F\subset \K[\x]$ be finite. Let $\T := \{T_1,\ldots,T_e\}$
be a finite set of regular chains of $\K[\x]$.
We call $\T$ a {\em Kalkbrener triangular decomposition} of $V(F)$
if we have $V(F)=\cup_{i=1}^e \overline{W(T_i)}$.
We call $\T$ a {\em Lazard-Wu triangular decomposition} of $V(F)$
if we have $V(F)=\cup_{i=1}^e W(T_i)$.

\begin{Proposition}[Th. 6.1. in~\cite{ALM97}]
\label{Proposition:prem}    %
Let $p$ and $T$ be respectively a polynomial and a regular chain
of $\K[\x]$. Then, 
$\prem{p, T}=0$ holds if and only if 
$p\in\sat{T}$ holds.
\end{Proposition}

\begin{Proposition}[Prop. 5 in~\cite{moreno00}]
\label{Proposition:equal-dim} %
Let $T$ and $T'$ be two regular chains of $\K[\x]$
such that $\sqrt{\sat{T}}\subseteq\sqrt{\sat{T'}}$ 
and $\dim{(\sat{T})}=\dim{(\sat{T'})}$ hold.
Let $p\in\K[\x]$ such that $p$ is regular w.r.t. $\sat{T}$. 
Then $p$ is also regular w.r.t. $\sat{T'}$.
\end{Proposition}

\begin{Proposition}[Prop. 4.4 in~\cite{ALM97}]
\label{Proposition:radsatprem}   %
Let $p\in\K[\x]$ and $T\subset\K[\x]$ be a regular chain. 
Let $v=\mvar{p}$ and $r=\prem{p, T_{\geq v}}$
such that $r\in\sqrt{\sat{T_{<v}}}$ holds.
Then, we have $p\in\sqrt{\sat{T}}$.
\end{Proposition}

\begin{Corollary}
\label{Corollary:sqrtsat}   %
Let $T$ and $T'$ be two regular chains of $\K[x_1,\ldots,x_{k}]$, 
where $1\leq k < n$. 
Let $p\in\K[\x]$ with $\mvar{p}=x_{k+1}$ such that 
$\init{p}$ is regular w.r.t. both $\sat{T}$ and $\sat{T'}$.
Assume that $\sqrt{\sat{T}} \,  \subseteq \,  \sqrt{\sat{T'}}$ holds. 
Then we also have $\sqrt{\sat{T\cup p}} \, \subseteq\ \, \sqrt{\sat{T'\cup p}}$.
\end{Corollary}

\begin{Proposition}[Lemma 4 in~\cite{CGLMP07}]
\label{Proposition:equivalence}     %
Let $p\in\K[\x]$. Let $T\subset\K[\x]$ be a regular chain. 
Then the following statements are equivalent:
\begin{itemize}
\item[$(i)$] the polynomial $p$ is regular w.r.t. $\sat{T}$,
\item[$(ii)$] for each prime ideal ${\p}$  associated 
              with $\sat{T}$, we have $p \not\in {\p}$,
\item[$(iii)$] the iterated resultant $\iter{p, T}$ is not zero.
\end{itemize}
\end{Proposition}

\begin{Corollary}
\label{Corollary:regular}  %
Let $p\in\K[\x]$ and $T\subset\K[\x]$ be a regular chain. 
Let $v := \mvar{p}$ and $r:=\iter{p, T_{\geq v}}$.
We have:
\begin{itemize}
\item[$(1)$]  the polynomial $p$ is regular w.r.t. $\sat{T}$ if and only if 
$r$ is regular w.r.t. $\sat{T_{<v}}$;
\item[$(2)$]  if $v\notin\mvar{T}$ and $\init{p}$ is
regular w.r.t. $\sat{T}$, then $p$ is regular w.r.t. $\sat{T}$.
\end{itemize}
\end{Corollary}

\section{Regular GCDs}    %
\label{sec:regulargcds}

As mentioned before,
Definition~\ref{Definition:regular-gcd} 
was introduced in~\cite{moreno00} as part of a formal
framework for algorithms manipulating 
regular chains~\cite{D5,Laz91a,ChGa92,Kalk93, YaZh91}.
In the present paper, the ring {\A}
will always be of the form $\K[\x]/\sqrt{\sat{T}}$.
Thus, a regular GCD of $p,t$ in $\A[y]$ is also called a regular GCD
of $p,t$ modulo $\sqrt{\sat{T}}.$

\begin{Definition}
\label{Definition:regular-gcd}
Let {\A} be a commutative ring with unity.
Let $p, t, g \in\A[y]$ with $t \neq 0$ and $g \neq 0$.
We say that $g\in\A[y]$ is a {\em regular \GCD} of $p, t$ if:
\begin{itemize}
\item[$(R_1)$] the leading coefficient of $g$ in $y$ is a regular element;
\item[$(R_2)$] $g$ belongs to the ideal generated by $p$ and $t$ 
               in ${\A}[y]$;
\item[$(R_3)$] if $\deg(g, y)>0$, then $g$ pseudo-divides both $p$ and $t$,
               that is, $\prem{p, g}=\prem{t, g}=0$.
\end{itemize}
\end{Definition}

\begin{Proposition}
\label{Proposition:regular-gcd}
For $1 \leq k \leq n$, 
let $T \subset \K[x_1,\ldots,x_{k-1}]$ be a regular chain, possibly empty.
Let $p,t,g\in\K[x_1,\ldots,x_{k}]$ be polynomials
with main variable $x_{k}$.
Assume $T\cup\{t\}$ is a regular chain and 
$g$ is a regular GCD of $p$ and $t$ modulo $\sqrt{\sat{T}}$.
We have:
\begin{itemize}
\item[$(i)$] if $\mdeg{g}=\mdeg{t}$, then $\sqrt{\sat{T\cup t}}=\sqrt{\sat{T\cup g}}$ and
$W(T\cup t) \ \subseteq \ Z(h_g, T\cup t) \, \cup \, W(T\cup g)$ both hold,
\item[$(ii)$] if $\mdeg{g}<\mdeg{t}$, let $q=\pquo{t,g}$, then $T\cup q$ is a regular chain and the following two relations hold:
\begin{itemize}
\item[$(ii.a)$] $\sqrt{\sat{T\cup t}} \ = \ \sqrt{\sat{T\cup g}} \, \cap \, \sqrt{\sat{T\cup q}}$,
\item[$(ii.b)$] $W(T\cup t) \ \subseteq \ Z(h_g, T\cup t) \, \cup \, W(T\cup g)\cup W(T\cup q),$
\end{itemize}
\item[$(iii)$] $W(T\cup g) \ \subseteq \ V(p)$,
\item[$(iv)$] $Z(p, T\cup t) \ \subseteq \ W(T\cup g) \, \cup \, Z(\{p, h_g\}, T\cup t)$.
\end{itemize}
\end{Proposition}

\begin{proof}
We first establish a relation between $p$, $t$ and $g$.
By definition of pseudo-division, there exist polynomials $q, r$ 
and a nonnegtive integer $e_0$ such that 
\begin{equation}
\label{eqs:tgq0}
h_g^{e_0}t = qg + r \ \ {\rm and}  \ \ r\in\sqrt{\sat{T}}
\end{equation}
both hold.
Hence, there exists an integer $e_1\geq0$ such that: 
\begin{equation}
\label{eqs:tgq}
(h_T)^{e_1}(h_g^{e_0}t-qg)^{e_1} \ \in \ \langle T \rangle
\end{equation}
holds, which implies: $t\in\sqrt{\sat{T\cup g}}$.
We first prove $(i)$. Since $\mdeg{t}=\mdeg{g}$ holds, 
we have $q \in \K[x_1,\ldots,x_{k-1}]$, and thus we have
$h_g^{e_0} \, h_t \ = \ q \, h_g$.
Since $h_t$ and $h_g$ are regular modulo $\sat{T}$,
the same property holds for $q$.
Together with~(\ref{eqs:tgq}), we obtain $g\in\sqrt{\sat{T\cup t}}$.
Therefore $\sqrt{\sat{T\cup t}}=\sqrt{\sat{T\cup g}}$.
The inclusion relation in $(i)$ follows from~(\ref{eqs:tgq0}). 

We prove $(ii)$. Assume $\mdeg{t}>\mdeg{g}$.
With~(\ref{eqs:tgq0}) and~(\ref{eqs:tgq}), this hypothesis implies
that $T \cup q$ is a regular chain and 
$t\in\sqrt{\sat{T\cup q}}$ holds. 
Since $t\in\sqrt{\sat{T\cup g}}$ also holds, 
$\sqrt{\sat{T\cup t}}$ is contained
in $\sqrt{\sat{T\cup g}} \, \cap \, \sqrt{\sat{T\cup q}}$.
Conversely, 
for any $f \in \sqrt{\sat{T\cup g}} \, \cap \, \sqrt{\sat{T\cup q}}$, 
there exists an integer $e_2\geq 0$ and $a \in {\K}[\x]$ such that 
$(h_gh_q)^{e_2}f^{e_2}-aqg\in\sat{T}$ holds.
With~(\ref{eqs:tgq0}) we deduce that 
$f\in\sqrt{\sat{T\cup t}}$ holds and so does $(ii.a)$.
With~(\ref{eqs:tgq0}), we have $(ii.b)$ holds.

We prove $(iii)$ and $(iv)$. 
Definition~\ref{Definition:regular-gcd} implies: 
$\prem{p, g} \, \in \, \sqrt{\sat{T}}$. Thus 
$p \in \sqrt{\sat{T \cup g}}$ holds, that is, 
$\overline{W(T \cup g)} \, \subseteq \, V(p)$,
which implies $(iii)$. 
Moreover, since $g\in\langle p, t, \sqrt{\sat{T}} \rangle$, 
we have $Z(p, T\cup t) \, \subseteq \, V(g)$, so we deduce $(iv)$.
\end{proof}

Let $p, t$ be two polynomials of $\K[x_1,\ldots,x_k]$, for $k \geq 1$.
Let $m=\deg(p, x_k)$, $n=\mdeg{t, x_k}$.
Assume that $m,n\geq 1$.
Let $\lambda=\min{m,n}$.
Let $T$ be a regular chain of $\K[x_1,\ldots,x_{k-1}]$.
Let $\B=\K[x_1,\ldots,x_{k-1}]$ and $\A=\B / \sqrt{\sat{T}}$.

Let $S_0,\ldots,S_{\lambda-1}$ be the subresulant polynomials~\cite{Mis93, Ducos00} 
of $p$ and $t$ w.r.t. $x_k$ in $\B[x_k]$.
Let $s_i=\coeff{S_i, x_k^{i}}$ be the principle
subresultant coefficient of $S_i$, for $0 \leq i \leq \lambda-1 $.
If $m\geq n$, we define $S_{\lambda}=t$, $S_{\lambda+1}=p$, 
$s_{\lambda}=\init{t}$ and $s_{\lambda+1}=\init{p}$.
If $m<n$, we define $S_{\lambda}=p$, $S_{\lambda+1}=t$, 
$s_{\lambda}=\init{p}$ and $s_{\lambda+1}=\init{t}$.

The following theorem provides sufficient conditions for 
$S_j$ (with $1\leq  j \leq \lambda+1$) 
to be a regular GCD of $p$ and $t$ in $\A[x_k]$.
\begin{Theorem}
\label{Theorem:regulargcd}
Let $j$  be an integer, with $1\leq  j \leq \lambda+1$,
such that $s_j$ is a regular element of $\A$
and such that for any $0\leq i< j$, we have $s_i=0$ in $\A$.
Then $S_j$ is a regular GCD of $p$ and $t$ in $\A[x_k]$.
\end{Theorem}
\begin{proof}
By  Definition~\ref{Definition:regular-gcd},
 it suffices to 
prove that both $\prem{p, S_j, x_k}=0$ and $\prem{t, S_j, x_k}=0$
hold in $\A$.
By symmetry we only prove the former equality.

Let ${\p}$ be any prime ideal associated with $\sat{T}$.
Define $\DD=\K[x_1,\ldots,x_{k-1}]/\p$ and let $\LL$ be 
the fraction field of the integral domain $\DD$.
Let $\phi$ be the homomorphism from $\B$ to $\LL$.
By Theorem~\ref{Theorem:gcd} of Appendix~\ref{app:subresultantchain}, we know that $\phi(S_j)$
is a GCD of $\phi(p)$ and $\phi(t)$ in $\LL[x_k]$.
Therefore there exists a polynomial $q$ of $\LL[x_k]$
such that $p=qS_j$ in $\LL[x_k]$, which implies that 
there exists a nonzero element $a$ of $\DD$ 
and a polynomial $q'$ of $\DD[x_k]$ such that 
$ap=q'S_j$ in $\DD[x_k]$.
Therefore $\prem{ap, S_j}=0$ in $\DD[x_k]$, 
which implies that $\prem{p, S_j}=0$ in $\DD[x_k]$.
Therefore $\prem{p, S_j}$ belongs to $\p$ and
thus to $\sqrt{\sat{T}}$.
So $\prem{p, S_j, x_k}=0$ in $\A$.
\end{proof}

\section{The incremental algorithm}  %
\label{sec:incremental}

\begin{figure*}
\centering
{\small
\begin{tabular}{|l|c|c|c|c|c||c|c|c|c|c|}
\hline
 & sys                     & \multicolumn{4}{c||}{Input size} & \multicolumn{5}{c|}{Output size}\\\hline
 &  & \#v & \#e & deg & dim& GL & GS & GD & TL & TK \\\hline
1& 4corps-1parameter-homog & 	  4 & 3 & 8 & 1 & -  & -  &21863& - & 30738 \\
2& 8-3-config-Li & 	  12 & 7 & 2 & 7 & 67965 & -  &72698&7538& 1384 \\
3&Alonso-Li & 	  7 & 4 & 4 & 3 & 1270 & -  &614&2050& 374 \\
4&Bezier & 	  5 & 3 & 6 & 2 & -  & -  &32054& - & 114109 \\
5&Cheaters-homotopy-1 & 	  7 & 3 & 7 & 4 & 26387452 & -  &17297& - & 285 \\
7&childDraw-2 & 	  10 & 10 & 2 & 0  & 938846 & -  &157765& - & -  \\
8&Cinquin-Demongeot-3-3 & 	  4 & 3 & 4 & 1 & 1652062 & -  &680&2065& 895 \\
9&Cinquin-Demongeot-3-4 & 	  4 & 3 & 5 & 1 & -  & -  &690& - & 2322 \\
10&collins-jsc02 & 	  5 & 4 & 3 & 1 & -  & -  &28720&2770& 1290 \\
11&f-744 & 	  12 & 12 & 3 & 1 & 102082 & -  &83559&4509& 4510 \\
12&Haas5 & 	  4 & 2 & 10 & 2 & -  & -  &28& - & 548 \\
14&Lichtblau & 	  3 & 2 & 11 & 1 & 6600095 & -  &224647&110332& 5243 \\
16&Liu-Lorenz & 	  5 & 4 & 2 & 1 & 47688 & 123965 &712&2339& 938 \\
17&Mehta2 & 	  11 & 8 & 3 & 3 & -  & -  &1374931&5347& 5097 \\
18&Mehta3 & 	  13 & 10 & 3 & 3 & -  & -  & - &25951& 25537 \\
19&Mehta4 &  	  15 & 12  & 3 & 3 & -  & -  & - &71675& 71239 \\
21&p3p-isosceles & 	  7 & 3 & 3 & 4 & 56701 & -  &1453&9253& 840 \\
22&p3p & 	  8 & 3 & 3 & 5 & 160567 & -  &1768& - & 1712 \\
23&Pavelle & 	  8 & 4 & 2 & 4 & 17990 & -  &1552&3351& 1086 \\
24&Solotareff-4b & 	  5 & 4 & 3 & 1 & 2903124 & -  &14810&2438& 872 \\
25&Wang93 & 	  5 & 4 & 3 & 1 & 2772 & 56383 &1377&1016& 391 \\
26&Xia & 	  6 & 3 & 4 & 3 & 63083 & 2711 &672&1647& 441 \\
27&xy-5-7-2 & 	  6 & 3 & 3 & 3  & 12750 & - &599& - & 3267 \\

\hline
\end{tabular}
}
\newline
{\textbf{Table 1} The input and output sizes of systems}
\end{figure*}

In this section, we present an algorithm to compute 
Lazard-Wu triangular decompositions in an incremental manner. 
We recall the concepts of a {\em process} and a {\em regular (delayed) split}, 
which were introduced as Definitions 9 and 11 
in~\cite{moreno00}. 
To serve our purpose, we modify the definitions as below.

\begin{Definition}
\label{Definition:process}  %
A {\em process} of $\K[\x]$ is a pair $(p, T)$, 
where $p\in\K[\x]$ is a polynomial and 
$T\subset \K[\x]$ is a regular chain.
The process $(0, T)$ is also written as $T$ for short. 
Given two processes $(p, T)$ and $(p', T')$, 
let $v$ and $v'$ be respectively the greatest variable
appearing in $(p, T)$ and $(p', T')$.
We say $(p, T)\prec(p', T')$ if:
$(i)$ either $v<v'$;
$(ii)$ or $v=v'$ and $\dim{T}<\dim{T'}$;
$(iii)$ or $v=v'$, $\dim{T}=\dim{T'}$ and $T\prec T'$;
$(iv)$ or $v=v'$, $\dim{T}=\dim{T'}$, $T\sim T'$ and $p\prec p'$.
We write $(p, T)\sim (p', T')$ if neither $(p, T)\prec(p', T')$
nor $(p', T')\prec(p, T)$ hold.
Clearly any sequence of processes which is strictly decreasing w.r.t. $\prec$
is finite.
\end{Definition}

\begin{Definition}
\label{Definition:split}
Let $T_i$, $1\leq i\leq e$, be regular chains of $\K[\x]$.
Let $p\in\K[\x]$.
We call $T_1,\ldots,T_e$ a {\em regular split} of $(p, T)$ 
whenever we have 
\begin{itemize}
\item[$(L_1)$] $\sqrt{\sat{T}}\subseteq\sqrt{\sat{T_i}}$
\item[$(L_2)$] $W(T_i)\subseteq V(p)$ (or equivalently $p\in\sqrt{\sat{T_i}}$)
\item[$(L_3)$] $V(p)\cap W(T)\subseteq\cup_{i=1}^e W(T_i)$
\end{itemize}
We write as $(p, T)\longrightarrow T_1,\ldots,T_e$.
Observe that the above three conditions are equivalent to 
the following relation.
$$
V(p)\cap W(T)\subseteq W(T_1)\cup\cdots\cup W(T_e)\subseteq V(p)\cap\overline{W(T)}.
$$
Geometrically, this means that 
we may compute a little more than $V(p)\cap W(T)$; 
however,  $W(T_1)\cup\cdots\cup W(T_e)$
is a ``sharp'' approximation of the intersection
of  $V(p)$ and $W(T)$.
\end{Definition}

Next we list the specifications of our triangular decomposition algorithm
and its subroutines.
We denote by $R$ the polynomial ring $\K[\x]$, 
where $\x=x_1<\cdots<x_n$.

\smallskip\noindent{\Triangularize{F, R}}
\begin{itemize}
\item {\bf Input:}
$F$, a finite set of polynomials of $R$
\item {\bf Output:} 
A Lazard-Wu triangular decomposition of $V(F)$.
\end{itemize}
\Intersect{p, T, R}
\begin{itemize}
\item {\bf Input:}
$p$, a polynomial of $R$; 
$T$, a regular chain of $R$
\item {\bf Output:} a set of regular chains $\{T_1,\ldots, T_e\}$ such chat
$(p, T)\longrightarrow T_1,\ldots,T_e$.
\end{itemize}
\Regularize{p, T, R}
\begin{itemize}
\item {\bf Input:}
$p$, a polynomial of $R$;
$T$, a regular chain of $R$.
\item {\bf Output:} a set of pairs $\{[p_1, T_1],\ldots, [p_e, T_e]\}$ such that 
for each $i, 1\leq i\leq e$: 
$(1)$ $T_i$ is a regular chain;
$(2)$ $p=p_i~ {\rm mod}~ \sqrt{\sat{T_i}}$;
$(3)$ if $p_i=0$, then $p_i\in\sqrt{\sat{T_i}}$ 
   otherwise $p_i$ is regular modulo $\sqrt{\sat{T_i}}$; 
moreover we have $T\longrightarrow T_1,\ldots,T_e$.
\end{itemize}
\SubresultantChain{p, q, v, R}
\begin{itemize}
\item {\bf Input:}
$v$, a variable of $\{x_1,\ldots,x_n\}$;
$p$ and $q$, polynomials of $R$, whose main variables are both $v$.
\item {\bf Output:} a list of polynomials $(S_0,\ldots,S_{\lambda})$, 
where $\lambda={\rm min}(\mdeg{p}, \mdeg{q})$, such that 
$S_i$
 is the $i$-{th} subresultant of $p$ and $q$ w.r.t. $v$.
\end{itemize}
\RegularGcd{p, q, v, S, T, R}
\begin{itemize}
\item {\bf Input:} $v$, a variable of $\{x_1,\ldots,x_n\}$,
\begin{itemize}
\item $T$, a regular chain of $R$ such that $\mvar{T}<v$,
\item $p$ and $q$, polynomials of $R$ with the same main variable $v$ such that: 
$\init{q}$ is regular modulo $\sqrt{\sat{T}}$; 
\res{p,q,v} belongs to $\sqrt{\sat{T}}$,
\item $S$, the subresultant chain of $p$ and $q$ w.r.t. $v$.
\end{itemize}
\item {\bf Output:} a set of pairs $\{[g_1, T_1],\ldots,[g_e, T_e]\}$ such that
$T\longrightarrow T_1,\ldots,T_e$ and for each $T_i$: 
if $\dim{T}=\dim{T_i}$, 
then $g_i$ is a regular GCD of $p$ and $q$ modulo $\sqrt{\sat{T_i}}$;
otherwise $g_i=0$, which means undefined.
\end{itemize}
\IntersectFree{p, x_i, C, R}
\begin{itemize}
\item {\bf Input:}
$x_i$, a variable of $\x$; 
$p$, a polynomial of $R$ with main variable $x_i$;
$C$, a regular chain of $\K[x_1,\ldots,x_{i-1}]$.
\item {\bf Output:} a set of regular chains $\{T_1,\ldots, T_e\}$ 
such that $(p, C)\longrightarrow(T_1,\ldots,T_e)$.
\end{itemize}
\IntersectAlgebraic{p, T, x_i, S, C, R}
\begin{itemize}
\item {\bf Input:} $p$, a polynomial of $R$ with main variable $x_i$,
\begin{itemize}
\item $T$, a regular chain of $R$, where $x_i\in\mvar{T}$,
\item $S$, the subresultant chain of $p$ and $T_{x_i}$ w.r.t. $x_i$,
\item $C$, a regular chain of $\K[x_1,\ldots,x_{i-1}]$,
such that: $\init{T_{x_i}}$ is regular modulo $\sqrt{\sat{C}}$; 
the resultant of $p$ and $T_{x_i}$, which is $S_0$, belongs to $\sqrt{\sat{C}}$.
\end{itemize}
\item {\bf Output:}
a set of regular chains $T_1,\ldots,T_e$ such that 
$(p, C\cup T_{x_i})\longrightarrow T_1,\ldots,T_e$.
\end{itemize}
\CleanChain{C, T, x_i, R}
\begin{itemize}
\item {\bf Input:}
$T$, a regular chain of $R$;
$C$, a regular chain of $\K[x_1,\ldots,x_{i-1}]$ such that
$\sqrt{\sat{T_{<x_i}}}\subseteq\sqrt{\sat{C}}$.
\item {\bf Output:}
if $x_i\notin\mvar{T}$, return $C$; 
otherwise return a set of regular chains $\{T_1,\ldots, T_e\}$ such that 
$\init{T_{x_i}}$ is regular modulo each $\sat{T_j}$, 
$\sqrt{\sat{C}}\subseteq \sqrt{\sat{T_j}}$ and
$W(C)\setminus V(\init{T_{x_i}})\subseteq\cup_{j=1}^e W(T_j)$.
\end{itemize}
\Extend{C, T, x_i, R}
\begin{itemize}
\item {\bf Input:}
$C$, is a regular chain of $\K[x_1,\ldots,x_{i-1}]$. 
$T$, a regular chain of $R$ such that $\sqrt{\sat{T_{<x_i}}}\subseteq\sqrt{\sat{C}}$.
\item {\bf Output:}
a set of regular chains $\{T_1,\ldots, T_e\}$ of $R$ such that 
$W(C\cup T_{\geq x_{i}})\subseteq\cup_{j=1}^e W(T_j)$ and $\sqrt{\sat{T}}\subseteq \sqrt{\sat{T_j}}$.
\end{itemize}

Algorithm {\sf SubresultantChain}
is standard, see~\cite{Ducos00}.
The algorithm {\sf Triangularize} is  a {\em principle algorithm}
which was first presented in~\cite{moreno00}.
We use the following conventions in our pseudo-code:
the keyword {\bf return} yields a result 
and terminates the current function call
while the keyword {\bf output} yields a result 
and keeps executing the current function call.

\begin{algorithm}
\linesnumbered
\caption{\Intersect{p, T, R}\label{Algo:Intersect}}
{\bf if} $\prem{p,T}=0$ {\bf then} return $\{T\}$\;
{\bf if} $p\in\K$ {\bf then} return $\{~\}$\;
$r := p$; $P := \{r\}$; $S := \{~\}$\; 
\While{$\mvar{r} \in \mvar{T}$}{
   $v := \mvar{r}$; $src := \SubresultantChain{r, T{_v}, v, R}$\;
   $S := S\cup\{src\}$; $r := \resultant{src}$\;
   {\bf if} $r=0$ {\bf then} break\;
   {\bf if} $r \in \K$ {\bf then} return $\{~ \}$\;
   $P := P\cup\{r\}$
}
$\T := \{\varnothing\}$; $\T' := \{~\}$; $i := 1$\;
\While{$i\leq n$}{
\For{$C\in\T$}{
   \uIf{$x_i \notin \mvar{P}$ {\bf and} $x_i \notin \mvar{T}$}{
        $\T' := \T'\cup\CleanChain{C, T, x_{i+1}, R}$
   }
   \uElseIf{$x_i \notin \mvar{P}$}{
            $\T' := \T'\cup\CleanChain{C\cup T_{x_i}, T, x_{i+1}, R}$
   }
   \uElseIf{$x_i\notin\mvar{T}$}{
       \For{$D \in \IntersectFree{P_{x_i}, x_i, C, R}$}{
            $\T' := \T'\cup\CleanChain{D, T, x_{i+1}, R}$
       }
   }
   \Else{
       {\small
       \For{$D\in\IntersectAlgebraic{P_{x_i}, T, x_i, S_{x_i}, C, R}$}{
            $\T' := \T'\cup\CleanChain{D, T, x_{i+1}, R}$
       }
       }
   }
}
   $\T := \T'$; $\T' := \{~\}$; $i := i+1$
}
   return $\T$
\end{algorithm}
\begin{algorithm}
\caption{\RegularGcd{p, q, v, S, T, R}}
    $\T := \{(T, 1)\}$\; 
    \While{$\T\neq\emptyset$}{
        let $(C, i)\in \T$; $\T := \T\setminus\{(C, i)\}$\;
        \For{$[f, D]\in\Regularize{s_i, C, R}$}{
             \lIf{$\dim{D}<\dim{C}$}{
                  output $[0, D]$
             }\;
             \lElseIf{$f=0$}{
                       $\T := \T\cup\{(D, i+1)\}$
             }\;
             \lElse{
                   output $[S_i, D]$
             }
        }
     }
\end{algorithm}
\begin{algorithm}
\linesnumbered
\caption{\IntersectFree{p, x_i, C, R}\label{Algo:IntersectFree}}
    \For{$[f, D] \in \Regularize{\init{p}, C, R}$}{
         \lIf{$f=0$}{
              output $\Intersect{\tail{p}, D, R}$
            }\;
          \Else{
              output $D\cup p$\;
              \For{$E\in\Intersect{\init{p}, D, R}$}{
                   output $\Intersect{\tail{p}, E, R}$
              }
         }
    }
\end{algorithm}
\begin{algorithm}
\linesnumbered
\caption{\IntersectAlgebraic{p, T, x_i, S, C, R}\label{Algo:IntersectAlgebraic}}
    \For{$[g, D]\in\RegularGcd{p, T_{x_i}, x_i, S, C, R}$}{
        \eIf{$\dim{D}<\dim{C}$}{
             \For{$E\in\CleanChain{D, T, x_{i}, R}$}{
                  output $\IntersectAlgebraic{p, T, x_i, S, E, R}$ 
             }
           }{
             output $D\cup g$\;
             \For{$E \in \Intersect{\init{g}, D, R}$}{
                  \For{$F\in\CleanChain{E, T, x_{i}, R}$}{
                       output $\IntersectAlgebraic{p, T, x_i, S, F, R}$ 
                  }
             }  
        }
    }
\end{algorithm}
\begin{algorithm}
\linesnumbered
\caption{\Regularize{p, T, R}}\label{Regularize}
{\bf if} $p\in\K$ {\bf or} $T=\varnothing$ {\bf then} return $[p, T]$\;
$v := \mvar{p}$\;
\eIf{$v\notin \mvar{T}$}{
     \For{$[f, C] \in\Regularize{\init{p}, T, R}$}{
          \lIf{$f=0$}{
               output $\Regularize{\tail{p}, C, R}$\;
          }\;
          \lElse
          {
               output $[p, C]$\;
          }
     }
}{

     $src := \SubresultantChain{p, T_v, v, R}$; $r:=\resultant{src}$\;
     \For{$[f, C]\in\Regularize{r, T_{<v}, R}$}{
          \uIf{$\dim{C}<\dim{T_{<v}}$}{
              \For{$D\in\Extend{C, T,v, R}$}{
                   output $\Regularize{p, D, R}$
              }
          }
          \lElseIf{$f\neq 0$}{
                   output $[p, C\cup T_{\geq v}]$
          }\;
          \Else{
              \For{$[g, D]\in\RegularGcd{p, T_v, v, src, C, R}$}{
                  \eIf{$\dim{D}<\dim{C}$}{
                       \For{$E\in\Extend{D, T, v, R}$}{
                            output $\Regularize{p, E, R}$\;
                       }
                     }{
                       {\bf if} $\mdeg{g}=\mdeg{T_v}$ {\bf then} output $[0, D\cup T_{\geq v}]$; next\;
                       output $[0, D\cup g\cup T_{>v}]$\;
                       $q := \pquo{T_v, g}$\;
                       output $\Regularize{p, D\cup q\cup T_{>v}, R}$\;
                       \For{$E\in\Intersect{h_g, D, R}$}{
                            \For{$F\in\Extend{E, T,v, R}$}{
                                 output $\Regularize{p, F, R}$
                            }
                       }
                  }
              }
          }
     }
}
\end{algorithm}
\begin{algorithm}
\linesnumbered
\caption{\Extend{C, T, x_i, R}}
    {\bf if} $T_{\geq x_i}=\varnothing$ {\bf then} return $C$\;
    let $p\in T$ with greatest main variable; $T':=T\setminus\{p\}$\;
    \For{$D\in\Extend{C, T', x_i, R}$}{
         \For{$[f, E]\in\Regularize{\init{p}, D}$}{
              {\bf if} $f\neq 0$ {\bf then} output $E\cup p$\;

         }            
    }
\end{algorithm}
\begin{algorithm}
\linesnumbered
\caption{\CleanChain{C, T, x_i, R}}
    {\bf if} $x_i\notin\mvar{T}$ {\bf or} $\dim{C}=\dim{T_{<x_i}}$ {\bf then} return $C$\;
    \For{$[f, D] \in \Regularize{\init{T_{x_i}}, C, R}$}{
         {\bf if} $f\neq 0$ {\bf then} output $D$
    }
\end{algorithm}
\begin{algorithm}
\caption{\Triangularize{F, R}\label{Algo:Triangularize}}
{\bf if} $F=\{~\}$ {\bf then} return $\{\varnothing\}$\;
Choose a polynomial $p\in F$ with maximal rank\;
\For{$T\in\Triangularize{F\setminus\{p\}, R}$}{
     output $\Intersect{p, T, R}$
}
\end{algorithm}

\section{Proof of the algorithms}  %
\label{Section:proof}
\begin{figure}
{
\centering
\includegraphics[width=0.45\textwidth]{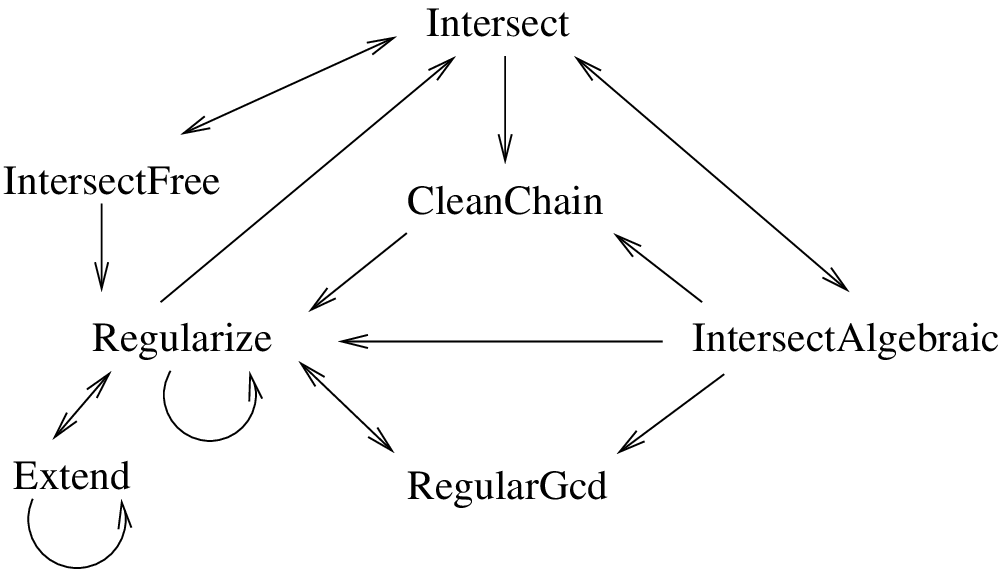}
\label{figure:algorithm}
\caption{Flow graph of the Algorithms}
}
\end{figure}

\begin{Theorem}
All the algorithms in Fig. 1 terminate.
\end{Theorem}
\begin{proof}
The key observation is  that 
the flow graph of  Fig. 1 can be transformed 
into an equivalent flow graph satisfying the following properties:
(1) the algorithms {\sf Intersect} and {\sf Regularize}
only call each other or themselves;
(2) all the other algorithms only call either 
  {\sf Intersect} or {\sf Regularize}.
Therefore, it suffices to show that
{\sf Intersect} and {\sf Regularize} terminate.

Note that the input of both functions is a process, say $(p, T)$.
One can check that, while executing a call with $(p, T)$ as input,
any subsequent call to either functions 
{\sf Intersect} or {\sf Regularize} will take 
a process $(p', T')$ as input such that $(p', T') \prec (p, T)$ holds.
Since a descending chain of processes is necessarily finite,  
both algorithms terminate.
\end{proof}

Since all algorithms terminate, and following the
flow graph of Fig. 1, each call to one of our algorithms
unfold to a finite dynamic acyclic graph (DAG) where
each vertex is a call to one of our algorithms.
Therefore, proving the
correctness of these algorithms
reduces to prove the following two points.
\begin{itemize}
\item {\em Base:} each algorithm call, 
which makes no subsequent calls to another algorithm or to itself,
is correct.
\item {\em Induction:} each algorithm call, 
which makes subsequent calls to another algorithm or to itself,
is correct, as soon as all subsequent calls are themselves correct.
\end{itemize}
For all algorithms in  Fig. 1, proving the base cases is straightforward.
Hence we focus on the induction steps.

\begin{Proposition}
\label{Proposition:IntersectFree}
${\sf IntersectFree}$ satisfies its specification.
\end{Proposition}
\begin{proof}
We have the following two key observations:
\begin{itemize}
\item $C\longrightarrow D_1,\ldots,D_s$, where $D_i$
are the regular chains in the output of ${\sf Regularize}$.
\item $V(p)\cap W(D)=W(D, p)\cup V(\init{p}, \tail{p})\cap W(D).$
\end{itemize}
Then it is not hard to conclude that $(p, C)\longrightarrow T_1,\ldots,T_e$.
\end{proof}

\begin{Proposition}
\label{Proposition:IntersectAlgebraic}
${\sf IntersectAlgebraic}$ is correct.
\end{Proposition}
\begin{proof}
We need to prove: $(p, C\cup T_{x_i})\longrightarrow T_1,\ldots,T_e$.
Let us prove $(L_1)$ now,
that is, for each regular chain $T_j$ in the output, we have
$\sqrt{\sat{C\cup T_{x_i}}}\subseteq \sqrt{\sat{T_j}}$.
First by the specifications of the called functions,
we have $\sqrt{\sat{C}}\subseteq\sqrt{\sat{D}}\subseteq \sqrt{\sat{E}}$, 
thus,
$\sqrt{\sat{C\cup T_{x_i}}}\subseteq\sqrt{\sat{E\cup T_{x_i}}}$
by Corollary~\ref{Corollary:sqrtsat}, since
$\init{T_{x_i}}$ is regular modulo both $\sat{C}$ and $\sat{E}$.
Secondly, since $g$ is a regular GCD of $p$ and $T_{x_i}$ modulo $\sqrt{\sat{D}}$,
we have $\sqrt{\sat{C\cup T_{x_i}}}\subseteq\sqrt{\sat{D\cup g}}$ by 
Corollaries~\ref{Corollary:sqrtsat} and Proposition~\ref{Proposition:regular-gcd}.

Next we prove $(L_2)$.
It is enough to prove that $W(D\cup g)\subseteq V(p)$ holds. 
Since $g$ is a regular GCD of $p$ and $T_{x_i}$ modulo 
$\sqrt{\sat{D}}$, 
the conclusion follows from point $(iii)$ of 
Proposition~\ref{Proposition:regular-gcd}.

Finally we prove $(L_3)$, that is 
$Z(p, C\cup T_{x_i})\subseteq\bigcup_{j=1}^e W(T_j)$.
Let $D_1,\ldots,D_s$ be the regular chains returned 
from Algorithm ${\sf RegularGcd}$.
We have $C\longrightarrow D_1,\ldots,D_s$,
which implies $Z(p, C\cup T_{x_i})\subseteq\cup_{j=1}^e  Z(p, D_j\cup T_{x_i})$.
Next since $g$ is a regular GCD of $p$ and $T_{x_i}$ modulo $\sqrt{\sat{D_j}}$, 
the conclusion follows from point $(iv)$ of 
Proposition~\ref{Proposition:regular-gcd}.
\end{proof}

\begin{Proposition}
\label{Proposition:Intersect}
${\sf Intersect}$ satisfies its
specification.
\end{Proposition}
\begin{proof}
The first while loop can be seen as a projection process. 
We claim that it produces a nonempty triangular set $P$ such that
$V(p)\cap W(T)=V(P)\cap W(T)$. The claim holds before staring the 
while loop. For each iteration, let $P'$ be the set of polynomials
obtained at the previous iteration. 
We then compute a polynomial $r$, which is the resultant of 
a polynomial in $P'$ and a polynomial in $T$. 
So $r\in\langle P', T\rangle$. By induction, 
we have $\langle p, T\rangle=\langle P, T\rangle$. So the claim holds.

Next, we claim that the elements in $\T$ satisfy the following invariants:
at the beginning of the $i$-{th} iteration of the second while loop, we have
\begin{itemize}
\item[$(1)$] each $C\in \T$ is a regular chain; if $T_{x_i}$ exists, then
      $\init{T_{x_i}}$ is regular modulo $\sat{C}$, 
\item[$(2)$] for each $C\in \T$, we have $\sqrt{\sat{T_{< x_i}}}\subseteq\sqrt{\sat{C}}$,
\item[$(3)$] for each $C\in \T$, we have $\overline{W(C)}\subseteq V(P_{< x_i})$,
\item[$(4)$] $
V(p)\cap W(T) \subseteq \bigcup_{C\in \T} Z(P_{\geq x_i},C\cup T_{\geq x_i}).
$
\end{itemize}
When $i=n+1$,  
we then have $\sqrt{\sat{T}}\subseteq \sqrt{\sat{C}}$, 
$W(C)\subseteq V(P)\subseteq V(p)$ for each $C\in\T$ and
$
V(p)\cap W(T)\subseteq \cup_{C\in \T} W(C).
$
So $(L_1), (L_2), (L_3)$ of Definition~\ref{Definition:split}
all hold. 
This concludes the correctness of the algorithm. 

Now we prove the above claims $(1)$, $(2)$, $(3)$, $(4)$ by induction. 
The claims clearly hold when $i=1$ since 
$C=\varnothing$ and 
$V(p)\cap W(T)=V(P)\cap W(T)$. 
Now assume that the loop invariants hold at the beginning 
of the $i$-th iteration.
We need to prove that it still holds at the beginning of the $(i+1)$-th iteration.
Let $C\in\T$ be an element picked up at the beginning of $i$-{th} iteration and 
let $L$  be the set of the new elements of $\T'$ generated from $C$.

Then for any $C'\in L$, 
claim $(1)$ clearly holds by specification of ${\sf CleanChain}$.
Next we prove $(2)$.
\begin{itemize}
\item if $x_i\notin\mvar{T}$, then $T_{<x_{i+1}}=T_{<x_i}$. 
By induction and specifications of called functions, 
we have $$\sqrt{\sat{T_{<x_{i+1}}}}\subseteq\sqrt{\sat{C}}\subseteq\sqrt{\sat{C'}}.$$
\item if $x_i\in\mvar{T}$, by induction we have 
$\sqrt{\sat{T_{<x_{i}}}}\subseteq\sqrt{\sat{C}}$ and $\init{T_{x_i}}$ 
is regular modulo both $\sat{C}$ and $\sat{T_{<x_{i}}}$. 
By Corollary~\ref{Corollary:sqrtsat} we have 
$$
\sqrt{\sat{T_{<x_{i+1}}}}\subseteq\sqrt{\sat{C\cup T_{x_i}}}\subseteq\sqrt{\sat{C'}}.
$$
\end{itemize}
Therefore $(2)$ holds. Next we prove claim $(3)$.
By induction and the specifications of called functions, 
we have $\overline{W(C')}\subseteq\overline{W(C\cup T_{x_i})}\subseteq V(P_{< x_{i}}).$
Secondly, we have $\overline{W(C')}\subseteq V(P_{x_i})$.
Therefore $\overline{W(C')}\subseteq V(P_{< x_{i+1}})$, that is $(3)$ holds. 
Finally,
since 
$V(P_{x_i})\cap W(C\cup T_{x_i})\setminus V(\init{T_{x_{i+1}}}) \subseteq\cup_{C'\in L} W(C')$, 
we have 
$
Z(P_{\geq x_i},C\cup T_{\geq x_i}) \subseteq \cup_{C'\in L} Z(P_{\geq x_{i+1}},C'\cup T_{\geq x_{i+1}}),
$ 
which implies that $(4)$ holds.
This completes the proof.
\end{proof}

\begin{Proposition}
${\sf Regularize}$ satisfies its specification.
\end{Proposition}
\begin{proof}
If $v\notin\mvar{T}$, 
the conclusion follows directly from point $(2)$ of Corollary~\ref{Corollary:regular}.
From now on, assume $v\in\mvar{T}$. 
Let $\L$ be the set of pairs $[p', T']$ in the output. 
We aim to prove the following facts
\begin{itemize}
\item[$(1)$] each $T'$ is a regular chain,
\item[$(2)$] if $p'=0$, then $p$ is zero modulo $\sqrt{\sat{T'}}$, 
otherwise $p$ is regular modulo $\sat{T}$,
\item[$(3)$] we have $\sqrt{\sat{T}}\subseteq\sqrt{\sat{T'}}$,
\item[$(4)$] we have $W(T)\subseteq \cup_{T'\in\L} W(T')$.
\end{itemize}
Statement $(1)$ is due to Proposition~\ref{Proposition:equal-dim}.
Next we prove $(2)$.
First, when there are recursive calls, the conclusion is obvious.
Let $[f, C]$ be a pair in the output of ${\Regularize{r, T_{<v}, R}}$. 
If $f\neq0$, the conclusion follows directly from point $(1)$ of Corollary~\ref{Corollary:regular}.
Otherwise, let $[g, D]$ be a pair in the output of the algorithm ${\RegularGcd{p, T_v, v, src, C, R}}$. 
If $\mdeg{g}=\mdeg{T_v}$, then by the algorithm of ${\sf RegularGcd}$, $g=T_v$. 
Therefore we have $\prem{p, T_v}\in\sqrt{\sat{C}}$, 
which implies that $p\in\sqrt{\sat{C\cup T_{\geq v}}}$ by Proposition~\ref{Proposition:radsatprem}.

Next we prove $(3)$.
Whenever  {\sf Extend} is called, $(3)$ holds immediately.
Otherwise, 
let $[f, C]$ be a pair returned by ${\Regularize{r, T_{<v}, R}}$. 
When $f\neq 0$, since $\sqrt{\sat{T_{<v}}}\subseteq \sqrt{\sat{C}}$ holds, 
we conclude $\sqrt{\sat{T}}\subseteq \sqrt{\sat{C\cup T_{\geq v}}}$ 
by Corollary~\ref{Corollary:sqrtsat}. 
Let $[g, D] \in {\RegularGcd{p, T_v, v, src, C, R}}$.
Corollary~\ref{Corollary:sqrtsat} 
and point $(ii)$ of Proposition~\ref{Proposition:regular-gcd}
imply that $\sqrt{\sat{T}} \ \subseteq \ \sqrt{\sat{D\cup T_{\geq v}}}$, 
$\sqrt{\sat{T}} \ \subseteq \ \sqrt{\sat{D\cup g\cup T_{> v}}}$ 
together with 
$\sqrt{\sat{T}} \ \subseteq \ \sqrt{\sat{D\cup q\cup T_{> v}}}$
hold. Hence $(3)$ holds.

Finally by point $(ii.b)$ of Proposition~\ref{Proposition:regular-gcd}, 
we have $W(D\cup T_v)\subseteq Z(h_g, D\cup T_v)\cup W(D\cup g)\cup W(D\cup q)$.
So $(4)$ holds.
\end{proof}

\begin{Proposition}
{\sf Extend} satisfies its specification.
\end{Proposition}
\begin{proof}
It clearly holds when $T_{\geq x_i}=\varnothing$, which is the base case.
By induction and the specification of {\sf Regularize}, 
we know that $\sqrt{\sat{T'}}\subseteq\sqrt{\sat{E}}$.
Since $\init{p}$ is regular modulo both $\sat{T'}$ and $\sat{E}$, 
by Corollary~\ref{Corollary:sqrtsat}, we have 
$\sqrt{\sat{T}}\subseteq\sqrt{\sat{E\cup p}}$.
On the other hand, we have $W(C\cup T'_{\geq x_i})\subseteq\cup W(D)$
and $W(D)\setminus V(h_p) \subseteq\cup~ W(E)$. 
Therefore $W(C\cup T_{\geq x_i})\subseteq\cup_{j=1}^e W(T_j)$, 
where $T_1, \ldots, T_e$ are the regular chains in the output. 
\end{proof}

\begin{Proposition}
{\sf CleanChain} satisfies its specification.
\end{Proposition}
\begin{proof}
It follows directly from 
Proposition~\ref{Proposition:equal-dim}.
\end{proof}

\begin{Proposition}
{\sf RegularGcd} satisfies its specification.
\end{Proposition}
\begin{proof}
Let $[g_i, T_i]$, $i=1,\ldots,e$, be the output.
First from the specification of {\sf Regularize}, 
we have $T\longrightarrow T_1,\ldots,T_e$.
When $\dim{T_i}=\dim{T}$, by Proposition~\ref{Proposition:equal-dim}
and Theorem~\ref{Theorem:regulargcd}, $g_i$ is 
a regular GCD of $p$ and $q$ modulo $\sqrt{\sat{T}}$.
\end{proof}

\section{ Kalkbrener decomposition}
\label{sec:Kalkbrener}

\begin{figure*}
\centering
{\small
\begin{tabular}{|l|c|c|c|c|c|c|c|c||c|c|c|c|c|}
\hline
sys   & \multicolumn{8}{c||}{{\sf Triangularize}} & \multicolumn{5}{c|}{{\sf Triangularize} versus other solvers}\\\hline
                     & TK13       & TK14        & TK       & TL13     &TL14      &TL   &STK & STL & GL          & GS          & WS   &    TL          & TK \\\hline
1 & - & 241.7 & 36.9 &-&-& - & 62.8 & -  &   - &   - &-&  	 - & 	  36.9 \\
2 & 8.7 & 5.3 & 5.9 &29.7&24.1& 25.8 & 6.0 & 26.6 &   108.7 &   - &27.8 & 25.8 & 5.9\\
3 & 0.3 & 0.3 & 0.4 &14.0 &2.4& 2.1 & 0.4 & 2.2 &   3.4 &   - &7.9 & 2.1 & 0.4 \\
4 & - & - & 88.2 &-&-& - & - & 	 - &   - &   - &-&  	 - & 88.2 \\
5 & 0.4 & 0.5 & 0.7 &-&-& - & 451.8 &  	- &   2609.5 &   - &-&  	 - & 0.7 \\
7 & - & - & - &-&-& - & 1326.8 & 1437.1 &   19.3 &   - &-&  	 - &  	 - \\
8 & 3.2 & 0.7 & 0.6 &-&55.9 & 7.1 & 0.7 & 8.8 &   63.6 &   - &-& 7.1 & 0.6 \\
9 & 166.1 & 5.0 & 3.1 &-&-& - & 3.3 & - &   - &   - &-&  	 - & 3.1 \\
10 & 5.8 & 0.4 & 0.4 &-&1.5& 1.5 & 0.4 & 1.5 &   - &   - &0.8 & 1.5 & 0.4 \\
11 & - & 29.1 & 12.7 &-&27.7& 14.8 & 12.9 & 15.1 &   30.8 &   - &-& 14.8 & 12.7 \\
12 & 452.3 & 454.1 & 0.3 &-&-& - & 0.3 & - &   - &   - &-&  	 - & 0.3 \\
14 & 0.7 & 0.7 & 0.3 &801.7 &226.5 & 143.5 & 0.3 & 531.3 &   125.9 &   - &-& 143.5 & 0.3 \\
16 & 0.4 & 0.4 & 0.4 &4.7 &2.6 & 2.3 & 0.4 & 4.4 &   3.2 &   2160.1 &40.2 & 2.3 & 0.4 \\
17 & - & 2.1 & 2.2 &-&4.5 & 4.5 & 2.2 & 6.2 &   - &   - &5.7 & 4.5 & 2.2 \\
18 & - & 15.6 & 14.4 &-&126.2 & 51.1 & 14.5 & 63.1 &   - &   - &-& 51.1 & 14.4 \\
19 & - & 871.1 & 859.4 &-&1987.5 & 1756.3 & 859.2 & 1761.8 &   - &   - &-& 1756.3 & 859.4 \\
21 & 1.2 & 0.6 & 0.3 &-&1303.1 & 352.5 & 0.3 & - &   6.2 &   - &792.8 & 352.5 & 0.3 \\
22 & 168.8 & 5.5 & 0.3 &-&-& - & 0.3 & - &   33.6 &   - &-&  	 - & 0.3 \\
23 & 0.8 & 0.9 & 0.5 &-&10.3 & 7.0 & 0.4 & 12.6 &   1.8 &   - &-& 7.0  & 0.5 \\
24 & 1.5 & 0.7 & 0.8 &-&1.9 & 1.9 & 0.9 & 2.0 &   35.2 &   - &9.1 & 1.9 & 0.8 \\
25 & 0.5 & 0.6 & 0.7 &0.6 &0.8 & 0.8 & 0.8 & 0.9 &   0.2 &   1580.0 &0.8 & 0.8 & 0.7 \\
26 & 0.2 & 0.3 & 0.4 &4.0 &1.9 & 1.9 & 0.5 & 2.7 &   4.7 &   0.1 &12.5 & 1.9 & 0.4 \\
27 & 3.3 & 0.9 & 0.6 &-&-& - & 0.7 & -  &   0.3 &   - &-&  	 - & 0.6   \\
\hline
\end{tabular}
}
\newline
{\textbf{Table 2} Timings of Triangularize versus other solvers}
\end{figure*}

In this section, we adapt the Algorithm 
{\sf Triangularize} (Algorithm~\ref{Algo:Triangularize}), 
in order to compute efficiently a Kalkbrener
triangular decomposition. 
The basic technique we rely on follows from 
Krull's principle ideal theorem.
\begin{Theorem}
Let $F\subset\K[\x]$ be finite, with cardinality $\#(F)$.
Assume $F$ 
generates a proper ideal of $\K[\x]$.
Then, for any minimal prime ideal ${\p}$ associated with $\langle F \rangle$, 
the height of ${\p}$ is less than or equal to $\#(F)$.
\end{Theorem}
\begin{Corollary}
Let $\T$ be a Kalkbrener triangular decomposition
of $V(F)$.
Let $T$ be a regular chain of $\T$, 
the height of which is greater than $\#(F)$.
Then $\T\setminus\{T\}$ is also a Kalkbrener triangular decomposition
of $V(F)$.
\end{Corollary}
Based on this corollary, we prune the decomposition tree 
generated during the computation of a Lazard-Wu triangular decomposition 
and remove the computation branches in which
the  height of every generated regular chain is 
greater than the number of polynomials in $F$.

Next we explain how to implement this tree pruning
technique to the algorithms of Section~\ref{sec:incremental}.
Inside {\sf Triangularize}, define $A=\#(F)$ 
and pass it to every call to {\sf Intersect} in order to signal 
{\sf Intersect} to output only regular chains 
with height no greater than $A$.
Next, in the second while loop of {\sf Intersect}, 
for the $i$-th iteration, 
we pass the height $A-\#(T_{\geq x_{i+1}})$ to 
{\sf CleanChain}, {\sf IntersectFree} and {\sf IntersectAlgebraic}.

In {\sf IntersectFree}, we pass its input height $A$ to 
every function call. 
Besides,  Lines $5$ to $6$ are executed only if the height of $D$ is
strictly less than $A$, since otherwise we would obtain regular chains
of height greater than $A$.
In other algorithms, we apply similar strategies as in {\sf Intersect}
and {\sf IntersectFree}.

\section{Experimentation}
\label{sec:Implementation}
Part of the algorithms presented in this paper 
are implemented in {\Maple 14} while
all of them are present  in the current development
version of {\Maple}.
Tables 1 and 2 report on our comparison between
 {\sf Triangularize} and other {\sc Maple} solvers.
The notations used in these tables are defined  below.

\smallskip\noindent{\small \bf Notation for {\sf Triangularize}.}
We denote by TK and TL the latest 
implementation of {\sf Triangularize}
for computing, respectively, Kalkbrener and Lazard-Wu decompositions,
in the current version of {\Maple}.
Denote by TK14 and TL14 the corresponding 
implementation in {\Maple14}.
Denote by TK13, TL13 the implementation
based on the algorithm of~\cite{moreno00} in {\Maple13}.
Finally, STK and STL are versions of  TK and TL respectively,
enforcing that all computed regular chains are squarefree,
by means of the algorithms in Appendix~\ref{sec:squarefree}.

\smallskip\noindent{\small \bf Notation for the other solvers.}
Denote by GL, GS, GD, respectively the function 
{\sf Groebner:-Basis} (plex order), 
{\sf Groebner:-Solve}, 
{\sf Groebner:-Basis} (tdeg order)
in current beta version of \Maple. 
Denote by WS the function {\sf wsolve} 
of the package {\tt Wsolve}~\cite{DingkangWang}, 
which decomposes a variety as a union of quasi-components 
of Wu Characteristic Sets. 

The tests were launched
on a machine with Intel Core 2 Quad {\small CPU} (2.40{\small GHz}) 
and 3.0{\small Gb} total memory.
The time-out is set as $3600$ seconds.
The memory usage is limited to $60\%$ of total memory.
In both Table 1 and 2, the symbol ``-'' means either
time or memory exceeds the limit we set.

The examples are mainly in positive dimension
since other triangular decomposition algorithms
are specialized to dimension zero~\cite{DMSWX05a}.
All examples are in characteristic zero.

In Table 1, we provide characteristics of the input systems 
and the sizes of the output obtained by different solvers.
For each polynomial system $F\subset\Q[\x]$, 
the number of variables appearing in $F$,
the number of polynomials in $F$,
the maximum total degree of a polynomial in $F$,
the dimension of the algebraic variety $V(F)$ 
are denoted respectively by $\#v$, $\#e$, $\deg$, $\dim$.
For each solver, the size of its output is measured by the 
total number of characters in the output.
To be precise, let ``dec'' and ``gb'' be respectively the output of the
{\sf Triangularize} and {\sf Groebner} functions. 
The {\Maple} command we use are {\sf length(convert(map(Equations, dec, R), string))} and 
{\sf length(convert(gb, string))}.
From Table 1, it is clear that {\sf Triangularize}
produces much smaller output than commands based 
on Gr\"obner basis computations.

TK, TL, GS, WS (and, to some extent, GL)
can all be seen as 
polynomial system solvers in the sense of that
they provide equidimensional decompositions
where components are represented by triangular sets.
Moreover, they are implemented in {\Maple}
(with the support of efficient C code in the case
of GS and GL).
The specification of TK
are close to those of  GS while
TL is related to WS, though the triangular sets
returned by WS are not necessarily regular chains.

In Table 2, we provide the timings of different
versions of {\sf Triangularize} and other solvers. 
From this table, it is clear that the implementations of {\sf Triangularize}, 
based on the algorithms presented in this paper
(that is TK14, TL14, TK, TL) outperform
the previous versions (TK13, TL13),
based on~\cite{moreno00},
by several orders of magnitude.
We observe also that TK outperforms GS and GL
while TL outperforms WS.

\smallskip\noindent{\bf Acknowledgments.}
The authors would like to thank  the support of 
{\sc Maplesoft}, {\sc Mitacs} and {\sc Nserc} of Canada.

\appendix
\medskip

\section{Specialization properties of subresultant chains}
\label{app:subresultantchain}
Let $\A$ be a commutative ring with identity
and let $k\leq \ell$ be two positive integers.
Let $M$ be an $k\times \ell$ matrix with coefficients in $\A$.
Let $M_i$ be the square submatrix of $M$ consisting of 
the first $k-1$ columns of $M$ and the $i_{th}$ column of $M$, 
for $i=k\cdots \ell$.
Let $\det(M_i)$ be the determinant of $M_i$. 
We denote by $\dpol{M}$ the element of $\A[x]$, called
the determinant polynomial of $M$, given by 
$$
\det{M_k}x^{\ell-k}+\det{M_{k+1}}x^{\ell-k-1}+\cdots+\det{M_{\ell}}.
$$
Let $f_1(x),\ldots,f_k(x)$ be a set of polynomials of $\A[x]$.
Let $\ell=1+\max{\deg{f_1(x)},\ldots,\deg{f_k(x)}}$.
The matrix $M$ of $f_1,\ldots,f_k$ is defined by $M_{ij}=\coeff{f_i,x^{\ell-j}}$.

Let $f=a_mx^m+\cdots+a_0$, $g=b_nx^n+\cdots+b_0$ be 
two polynomials of $\A[x]$ with positive degrees $m$ and $n$.
Let $\lambda=\min{m,n}$.
For any $0\leq i<\lambda$, 
let $M$ be the matrix of the polynomials $x^{n-1-i}f,\ldots,xf,f,x^{m-1-i}g,\ldots,xg,g$.
We define the $i_{th}$ subresultant
of $f$ and $g$, denoted by $S_i(f, g)$, as
$$
\begin{array}{rcl}
S_i(f, g)&=&\dpol{x^{n-1-i}f,\ldots,xf,f,x^{m-1-i}g,\ldots,xg,g}\\
         &=&\dpol{M}.
\end{array}
$$
Note that $S_i(f,g)$ is a polynomial in $\A[x]$ with degree at most $i$.
Let $s_i(f,g)=\coeff{S_i(f,g),x^i}$ and call it the principle subresultant
coefficient of $S_i$.

Let $\B$ be a UFD. Let $\phi$ be a homomorphism
from $\A$ to $\B$, which induces naturally also a homomorphism from $\A[x]$
to $\B[x]$. Let $m'=\deg(\phi(f))$ and $n'=\deg(\phi(g))$.

\begin{Lemma}
\label{Lemma:sk}
For any integer $0\leq k < \lambda$, if $\phi(s_k) \neq 0$, then $\phi(a_m)$ and $\phi(b_n)$
does not vanish at the same time.
Moreover, we have both $\deg{\phi(f)}\geq k$
and $\deg{\phi(g)}\geq k$.
\end{Lemma}
\begin{proof}
Observe that
$$
s_k=\left|
\begin{array}{ccccc}
a_m & a_{m-1} & \cdots & a_0&\\
&\cdots&        &\cdots    &\\
      & a_m  & a_{m-1} &\cdots&a_k\\
b_n & b_{n-1} & \cdots & b_0&\\
&\cdots&        &\cdots    &\\
      & b_n  & b_{n-1} &\cdots&b_k\\
\end{array}
\right|.
$$
Therefore there exists $i\geq k, j\geq k$
such that $\phi(a_i)\neq 0$ and $\phi(b_j)\neq0$.
The conclusion follows.
\end{proof}

\begin{Lemma}
\label{Lemma:smn}
Assume that $\phi(s_0)=\cdots=\phi(s_{\lambda-1})=0$.
Then, if $m\leq n$, we have
\begin{itemize}
\item[$(1)$] if $\phi(a_m)\neq 0$ and $\phi(b_n)=\cdots=\phi(b_{m})=0$, then $\phi(g)=0$
\item[$(2)$] if $\phi(a_m)=0$ and $\phi(b_n)\neq 0$, then $\phi(f)=0$
\end{itemize}
Symmetrically, if $m>n$, we have
\begin{itemize}
\item[$(3)$] if $\phi(b_n)\neq 0$ and $\phi(a_m)=\cdots=\phi(a_{n})=0$, then $\phi(f)=0$
\item[$(4)$] if $\phi(b_n)=0$ and $\phi(a_m)\neq 0$, then $\phi(g)=0$
\end{itemize}
\end{Lemma}
\begin{proof}
We prove $(1)$ and $(2)$, whose correctness implies $(3)$ and $(4)$ by symmetry. 
Let $i=\lambda-1=m-1$, then we have
$$
S_{m-1}=\dpol{x^{n-m}f,\ldots,xf, f, g}.
$$
Therefore 
$$
s_{m-1}=\left|
\begin{array}{ccccc}
a_m & \cdots & a_0     &\\
    &\ddots  & \ddots  &\\
    &        & a_{m}    &a_{m-1}\\
b_n & \cdots & b_m      &b_{m-1}\\
\end{array}
\right|.
$$
So from $\phi(b_n)=\cdots=\phi(b_m)=0$ and $\phi(s_{m-1})=0$, 
we conclude that $\phi(b_{m-1})=0$. 
On the other hand, if $\phi(a_m)=0$ and $\phi(b_n)\neq 0$, then $\phi(a_{m-1})=0$.

Now let consider $S_{m-2}$.
We have
$$
s_{m-2}=\left|
\begin{array}{ccccc}
a_m & a_{m-1}  & \cdots     &a_0     &        \\
    &\ddots  &      &\ddots  &        \\
    &        &a_{m} &a_{m-1}  & a_{m-2}\\
b_n & \cdots    &b_{m-1}&b_{m-2} &        \\    
    & b_n &\cdots    &b_{m-1} & b_{m-2}\\
\end{array}
\right|.
$$
From $\phi(b_{m-1})=0$, we conclude that $\phi(b_{m-2})=0$.
From $\phi(a_{m-1})=0$, we conclude that $\phi(a_{m-2})=0$.

So on so forth, finally, 
if $\phi(a_m)\neq 0$ and $\phi(b_n)=\cdots=\phi(b_m)=0$, we deduce that $\phi(b_{i})=0$, for all $0\leq i\leq m-1$, 
which implies that $\phi(g)=0$; 
if $\phi(a_m)= 0$ and $\phi(b_n)\neq 0$, we deduce that $\phi(a_{m-1})=\cdots=\phi(a_0)=0$, 
which implies that $\phi(f)=0$.
\end{proof}

\begin{Lemma}
\label{Lemma:dpol}
Let $i$ be an integer such that $1\leq i <\lambda$.
Assume that $\phi(a_m)\neq 0$.
If $i\leq n'$, then we have
$$
\begin{array}{ll}
\phi(S_i)=\phi(a_m)^{n-n'}{\rm dpol}(&\!\!\!\!\!x^{n'-1-i}\phi(f),\ldots,x\phi(f), \phi(f), \\
                               & \!\!\!\!\!x^{m-1-i}\phi(g),\ldots,x\phi(g), \phi(g))
\end{array}
$$
\end{Lemma}
\begin{proof}
If $i\leq n'$, then $n-n'\leq n-i$.
Therefore we have 
$$
\begin{array}{rcl}
\phi(S_i) &=& \phi(\dpol{x^{n-1-i}f,\ldots,xf, f,x^{m-1-i}g,\ldots, xg, g})\\
          &=& \phi({\rm dpol}(x^{n-1-i}\phi(f),\ldots,x\phi(f), \phi(f), \\
          &&x^{m-1-i}\phi(g),\ldots,x\phi(g), \phi(g)))\\
          &=& \phi(a_m)^{n-n'}{\rm dpol}(x^{n'-1-i}\phi(f),\ldots,x\phi(f), \phi(f), \\
          &&x^{m-1-i}\phi(g),\ldots,x\phi(g), \phi(g)))
\end{array}
$$
Done.
\end{proof}

\begin{Theorem}
\label{Theorem:gcd}
We have the following relations between 
the subresultants and the {\small GCD}
of $\phi(f)$ and $\phi(g)$: 
\begin{enumerate}
\item Let $k$, $0\leq k<\lambda$, be an integer such that
$\phi(s_k)\neq 0$ and for any $i$, $0\leq i<k$, 
$\phi(s_i)=0$. Then ${\rm gcd}(\phi(f), \phi(g))=\phi(S_k)$.
\item Assume that $\phi(s_i)=0$ for all $0\leq i<\lambda$. 
we have the following cases
\begin{enumerate}
\item if  $m\leq n$ and $\phi(a_m)\neq0$, then
${\rm gcd}(\phi(f), \phi(g))=\phi(f)$; symmetrically,
if $m>n$ and $\phi(b_n)\neq0$, then we have
${\rm gcd}(\phi(f), \phi(g))=\phi(g)$
\item if $m\leq n$ and $\phi(a_m)=0$ but $\phi(b_n)\neq0$, 
then we have ${\rm gcd}(\phi(f), \phi(g))=\phi(g)$; 
symmetrically, if $m\geq n$ and $\phi(b_n)=0$ but $\phi(a_m)\neq0$,
then we have ${\rm gcd}(\phi(f), \phi(g))=\phi(f)$
\item if $\phi(a_m)=\phi(b_n)=0$, then 
$${\rm gcd}(\phi(f), \phi(g))={\rm gcd}(\phi(\red{f}), \phi(\red{g}))$$
\end{enumerate}
\end{enumerate}
\end{Theorem}

\begin{proof}
Let us first prove $(1)$.
W.l.o.g, we assume $\phi(a_m)\neq 0$.
From Lemma~\ref{Lemma:sk}, we know that $k\leq n'$.
Therefore for $i\leq k$, we have $i\leq n'$.
By Lemma~\ref{Lemma:dpol}, 
$$
\begin{array}{rcl}
\phi(S_i) &=& \phi(a_m)^{n-n'}{\rm dpol}(x^{n'-1-i}\phi(f),\ldots,x\phi(f), \phi(f),\\
          &&x^{m-1-i}\phi(g),\ldots,x\phi(g), \phi(g))
\end{array}
$$
If $i<n'$, we have $\phi(S_i)=\phi(a_m)^{n-n'}S_i(\phi(f), \phi(g))$.
If $i=n'$, since $i<m$, we have 
$$
\begin{array}{rcl}
\phi(S_i) &=& \phi(a_m)^{n-n'}\dpol{x^{m-1-i}\phi(g),\ldots,x\phi(g), \phi(g)}\\
&=&\phi(a_m)^{n-n'}\phi(b_{n'})^{m-1-i}\phi(g).
\end{array}
$$
So for all $i<k$, we have $s_i(\phi(f), \phi(g))=0$.
If $k<n'$, we have $s_k(\phi(f), \phi(g))\neq0$. 
So $\gcd{\phi(f), \phi(g)}=\phi(S_k)$.
If $k=n'$, we have $\phi(b_{n'})=\phi(b_k)\neq 0$. 
Therefore $\gcd{\phi(f), \phi(g)}=\phi(g)=\phi(S_k)$.

Next we prove $(2a)$.
By symmetry, we prove it when $m\leq n$.
If $\phi(b_n)=\cdots=\phi(b_m)=0$, it follows directy from Lemma~\ref{Lemma:smn}.
Otherwise, we have $n'\geq m$.
By Lemma~\ref{Lemma:dpol}, for all $i<m$ we have
$$
\begin{array}{rcl}
\phi(S_i)&=&\phi(a_m)^{n-n'}{\rm dpol}(x^{n'-1-i}\phi(f),\ldots,x\phi(f), \phi(f), \\
         &&x^{m-1-i}\phi(g),\ldots,x\phi(g), \phi(g))
\end{array}
$$
That is $\phi(S_i)=\phi(a_m)^{n-n'}S_i(\phi(f), \phi(g))$.
Since $\phi(s_i)=0$, we deduce that $\phi(S_i)=\gcd{\phi(f), \phi(g)}$.

Finally $(2b)$ follows directly from Lemma~\ref{Lemma:smn} and $(2c)$ is obviouly true.
All done.
\end{proof}

\section{Squarefree decomposition}
\label{sec:squarefree}

Throughout this section, we assume that the coefficient field $\K$
is of characteristic zero.
We propose two strategies for computing a squarefree triangular decomposition.
The first one is a post-processing 
which applies Algorithm~\ref{Algo:Squarefree}
to every regular chain returned by Algorithm~\ref{Algo:Triangularize}.
The second consists of ensuring that, each output or intermediate 
regular chain generared during the execution
of Algorithm~\ref{Algo:Triangularize} is squarefree.

To implement the second strategy, we add an 
{\em squarefree option} to Algorithm~\ref{Algo:Triangularize}
and each of its subalgorithms.
If the option is set to {\em true}, this option requires that each output
regular chain is squarefree.
This is achieved by using Algorithm~\ref{Algo:Squarefree-poly}
whenever we need to construct new regular chains 
from a previous regular chain $T$ and a polynomial $p$
such that $T\cup p$ is known to be a regular chain.

\begin{algorithm}
\label{Algo:Squarefree-poly}
\linesnumbered
\caption{\Squarefree{p, x_i, T, R}\label{Algo:Squarefree}}
\KwIn{a polynomial ring $R=\K[x_1,\ldots,x_n]$, a variable\\
 $x_i$ of $R$,
a squarefree regular chain $T$ of $\K[x_1,\ldots,x_{i-1}]$,
a polynomial $p$ of $R$ with main variable $x_i$ such that
$T\cup p$ is a regular chain.
}
\KwOut{
a set of squarefree regular chains $T_1,\ldots,T_e$ such that 
$p\cup T\longrightarrow T_1,\ldots,T_e$.
}
       $p := \SquarefreePart{p}$\;
       {\bf if} $\mdeg{p}=1$ {\bf then} return $T\cup p$\;
       \Else{
          $src := \SubresultantChain{p,p',x_i, R}$\;
          return $\Squarefree{p, x_i, src, T, R}$\;
       }
\end{algorithm}

\begin{algorithm}
\linesnumbered
\caption{\Squarefree{p, x_i, src, T, R}\label{Algo:Squarefree-src}}
\KwIn{a polynomial ring $R=\K[x_1,\ldots,x_n]$, a variable\\ 
$x_i$ of $R$, a squarefree regular chain $T$ of $\K[x_1,\ldots,x_{i-1}]$,
a squarefree polynomial $p$ of $R$ with main variable $x_i$ such that
$T\cup p$ is a regular chain, 
the sub-resultant chain $src$ of $p$ and $p'$ w.r.t $x_i$.
}
\KwOut{a set of squarefree regular chains $T_1,\ldots,T_e$ such that 
$p\cup T\longrightarrow T_1,\ldots,T_e$.}
      $r := \resultant{src, R}$\;
      $\T := \{~\}$\;
      \For{$C\in\Regularize{r, T, R}$}{
           \lIf{$r\notin\sat{C}$}{
                output $C\cup p$; next\;
           }
           \Else{
                \If{$\dim{C}=\dim{T}$}{
                     $\T := \T\cup\{C\}$; next\;
                 }
                 \Else{
                     \For{$[f, D]\in\Regularize{\init{p}, C, R}$}{
                          {\bf if } $f\neq0$ {\bf then} $\T := \T\cup\{D\}$\;
                     }
                }
           }
      } 
      \While{$\T\neq\{~\}$}{
             let $C\in\T$; $\T:=\T\setminus\{C\}$\;
             \For{$[g, D]\in\RegularGcd{p, p', x_i, src, C, R}$}{ 
                  \eIf{$\dim{D}=\dim{C}$}{
                       output $D\cup\pquo{p,g}$\;
                       \For{$E\in\Intersect{\init{g}, D, R}$}{
                            \For{$[f, F]\in\Regularize{\init{p}, E, R}$}{
                                 {\bf if } $f\neq0$ {\bf then} $\T := \T\cup\{F\}$\;
                            }
                       }
                     }{
                       \For{$[f, E]\in\Regularize{\init{p}, D, R}$}{
                            {\bf if } $f\neq0$ {\bf then} $\T := \T\cup\{E\}$\;
                       }
                  }
             }
 
      }
\end{algorithm}

\begin{algorithm}
\linesnumbered
\caption{\Squarefree{T, R}\label{Algo:Squarefree}}
\KwIn{a polynomial ring $R=\K[x_1,\ldots,x_n]$, 
a regular chain $T$ of $R$.
}
\KwOut{
a set of squarefree regular chains $T_1,\ldots,T_e$ such that 
$T\longrightarrow T_1,\ldots,T_e$.
}
       $T := \{\SquarefreePart{p}\mid p\in T\}$\;
       $S := \{~\}$\;
       \For{$p\in T$}{
            \If{$\mdeg{p} > 1$}{$S:=S\cup\{\SubresultantChain{p,p',\mvar{p},R}\}$\;}
       }
       $\T := \{\varnothing\}$; $\T' := \{~\}$; $i := 1$\;
       \While{$i\leq n$}{
              \For{$C\in\T$}{
                   \eIf{$x_i \notin \mvar{T}$}{
                        $\T' := \T'\cup\CleanChain{C, T, x_{i+1}, R}$
                      }{
                        \eIf{$\mdeg{T_{x_i}}=1$}{
                             $\T' := \T'\cup\CleanChain{C\cup\{T_{x_i}\}, T, x_{i+1}, R}$
                        }{
                          \For{$D\in\Squarefree{T_{x_i}, x_i, S_{x_i}, C, R}$}{
                               $\T' := \T'\cup\CleanChain{D, T, x_{i+1}, R}$
                          }
                        }
                   }
              }
              $\T := \T'$; $\T' := \{~\}$; $i := i+1$\;
       }
       
   return $\T$
\end{algorithm}

\end{document}